\documentclass[letterpaper, 10 pt, conference]{ieeeconf}  

\IEEEoverridecommandlockouts                              
\usepackage{graphics} 
\usepackage{epsfig} 
\usepackage{amsmath} 
\usepackage{amssymb}  
\usepackage{cite}
\usepackage{color}
\usepackage[normalem]{ulem}
\usepackage{subcaption}

\newtheorem{theorem}{Theorem}
\newtheorem{lemma}[theorem]{Lemma}

\renewcommand{\AA}{\mathbb{A}}
\newcommand{\RR}{\mathbb{R}}
\newcommand{\XX}{\mathbb{X}}
\newcommand{\DD}{\mathbb{D}}
\newcommand{\VV}{\mathbb{V}}

\newcommand{\B}{\mathcal{B}}
\newcommand{\I}{\mathcal{I}}
\renewcommand{\L}{\mathcal{L}}
\renewcommand{\P}{\mathcal{P}}
\newcommand{\K}{\mathcal{K}}
\newcommand{\N}{\mathcal{N}}

\newcommand{\rh}{\hat{\rho}}

\title{\LARGE \bf
Controlled density transport using Perron Frobenius generators
}

\author{Jake Buzhardt and Phanindra Tallapragada
\thanks{The authors are with the Department of Mechanical Engineering, Clemson University, Clemson, SC, USA.
        {\tt\small jbuzhar@g.clemson.edu, ptallap@clemson.edu}}%
}

\begin{document}

\maketitle
\thispagestyle{empty}
\pagestyle{empty}

\begin{abstract}
We consider the problem of the transport of a density of states from an initial state distribution to a desired final state distribution through a dynamical system with actuation.  
In particular, we consider the case where the control signal is a function of time, but not space; that is, the same actuation is applied at every point in the state space. 
This is motivated by several problems in fluid mechanics, such as mixing and manipulation of a collection of particles by a global control input such as a uniform magnetic field, 
as well as by more general control problems where a density function describes an 
uncertainty distribution or a distribution of agents in a multi-agent system. 
We formulate this problem using the generators of the Perron-Frobenius operator associated with the drift and control vector fields of the system. 
By considering finite-dimensional approximations of these operators, the density transport problem can be expressed as a control problem for a bilinear system in a high-dimensional, lifted state.  
With this system, we frame the density control problem as a problem of driving moments of the density function to the moments of a desired density function, where the moments of the density can be expressed as an output which is linear in the lifted state. This output tracking problem for the lifted bilinear system is then solved using differential dynamic programming, an iterative trajectory optimization scheme.  

\end{abstract}

\section{Introduction}
In this paper, we consider the problem of controlled density transport, where given an initial distribution of states specified by a density function, we seek to determine a control sequence to drive this initial distribution to a desired final distribution.  
We consider the case where a common control signal is applied to the entire distribution of states.
This differs from the usual formulation of swarm control and optimal transport problems, where typically each agent can select a control input independently, making the control signal a function of the states and time. 
This problem of density transport is motivated by problems of manipulation of a large collection of agents using a uniform control signal \cite{becker_iros2013,bt_pre19}. 
The transport of density also has relevance to the propagation of an uncertainty distribution arising due to uncertainty in the initial state or of a model parameter through an otherwise deterministic control system (see, e.g. \cite{mesbah2014probcon,boutselis2019gpcddp,theodorou2021receding_gpcddp}). We formulate and solve this problem using an operator theoretic approach, specifically using the generator of the Perron-Frobenius operator.  

In recent years the operator theoretic approach to dynamical systems and control has gained significant research attention \cite{budivsic2012appliedkoopmanism, brunton22_modernkoopman,ottorowley2021koopmancontrol}. 
A dynamical system can be framed in terms of such an operator either by considering the evolution of observable functions of the state using the Koopman operator or by considering the evolution of densities of states using the Perron-Frobenius operator \cite{lasota_1994}. 
The interest in these approaches is primarily due to the fact that these operators allow for a linear, although typically infinite dimensional representation of a nonlinear system.  
The linearity of these operators is useful from an analytical perspective, as it allows for the use of linear systems techniques such as the analysis of eigenvalues and eigenfunctions, but also from a computational perspective, as in many cases a useful approximation for these operators can be found by considering a finite dimensional approximation in which the operator is represented as a matrix acting on coordinates corresponding to a finite set of a set of dictionary functions \cite{dellnitz1999approximation,williams2015data,Klus16_onthenumerical}.

In applications in control systems, much of the recent work has been on developing methods involving the Koopman operator \cite{kaiser2020DataDrivenOperatorsReview,ottorowley2021koopmancontrol}, as the transformation to a space of observable functions can be viewed as a nonlinear change of coordinates which maps the system to a higher dimensional space where the dynamics are (approximately) linear \cite{korda2018linear}.  This makes the numerical approximation of the operator particularly amenable to linear control methods, such as the linear quadratic regulator (LQR) and model predictive control (MPC) \cite{korda2018linear,proctor2018dmdc,williams_edmdc}.  On the other hand, the Perron-Frobenius operator propagates densities of states forward in time along trajectories of the system, which can have multiple interpretations in the controlled setting.  For example, the Perron-Frobenius operator and the Liouville equation, the related PDE formulation, have been used to determine controls for agents in an ensemble or swarm formulation \cite{brockett2012notes,grover2018_OTgenerators}. 
Such formulations are closely related to optimal transport problems which also involve driving an initial distribution to a desired final distribution (see, e.g., \cite{grover2018_OTgenerators,chen2021optimal,halder2014geodesic}).
Formulations involving the Perron-Frobenius operator have also been used in the context of fluid flows to study the transport of distributions of fluid particles and to detect invariant or almost invariant sets \cite{tallapragada2013set,Froyland2010transport}.

Our approach involves first obtaining a finite dimensional approximation of the Perron-Frobenius generators associated with the drift and control vector fields of the system, which allow us to represent the density transport dynamics as a bilinear system in a lifted state.  With this system, we frame the density control problem as a problem of driving moments of the density function to the moments of a desired density function, where the moments of the density can be expressed as an output which is linear in the lifted state. This output tracking problem for the lifted bilinear system is then solved using differential dynamic programming (DDP), an iterative trajectory optimization scheme.  


\section{Preliminaries}
Consider first the autonomous dynamical system on 
a measure space $(\XX \subset \RR^n, \mathcal{A}, \mu)$ with a $\sigma$-algebra $\mathcal{A}$ on $\XX$ and $\mu$ a measure on $(\XX, \mathcal{A})$,
\begin{equation}
    \dot{x} =  f(x) 
\end{equation} 
and denote the associated time-$t$ flow from an initial state $x_0$ as $\Phi^t(x_0)$, where $x\in\XX$ is the state. 
The  Perron-Frobenius operator $\P^t : L^1(\XX)\mapsto L^1(\XX)$
associated with the flow map $\Phi^t$ is defined  as 
\begin{equation}
    \int_{\AA} \left[\P^t \, \rho\right](x) \, dx = \int_{(\Phi^t)^{-1}(\AA)} \rho(x) \, dx
\end{equation}
for any   $\AA \in \mathcal{A}$, assuming that the relevant measure $\mu$ is absolutely continuous with respect to the Lebesgue measure and can thus be expressed in terms of a density $\rho$ (i.e., $d\mu(x) = \mu(dx) = \rho(x)dx$).  It can be shown that the family of these operators $\{\P^t\}_{t\geq0}$ form a semigroup, (see  \cite{lasota_1994}).  The generator of this semigroup is known as the Liouville operator, denoted $\L$, or Perron-Frobenius generator and expresses the deformation of the density $\rho$ under infinitesimal action of the operator $\P^t$ \cite{cvitanovic_chaosbook,lasota_1994}.  That is, 
\begin{equation} \label{eq:generator_continuity}
    \frac{d\rho}{dt} = \L\rho = -\nabla_x\cdot(\rho f)
\end{equation}
Alternatively, the action of the generator can be written in terms of the Perron-Frobenius operator as 
\begin{equation}\label{eq:generator_limit}
    \L\rho = \lim_{t\to 0} \frac{\P^{t}\rho - \rho}{t} = \lim_{t\to 0} \left(\frac{\P^t-\I}{t}\right)\rho
\end{equation}
where $\I$ is the identity operator. 

\begin{lemma} \label{lem_additivity}
Suppose the Liouville operator associated with a vector field $f_1:\XX \mapsto \RR^n$ is denoted by $\L_1$ and the Liouville operator associated with the vector field $f_2:\XX\mapsto \RR^n$ by $\L_2$, then the Liouville operator associated with the vector field $f(x) = f_1(x) + f_2(x)$, is $\L = \L_1 + \L_2$.
\end{lemma}
\begin{proof} 
The proof is a direct consequence of Eq. \ref{eq:generator_continuity}.  Suppose $f(x) = f_1(x) + f_2(x)$. Then $\L \rho = -\nabla_x \cdot (\rho (f_1 + f_2)) = -\nabla_x \cdot (\rho f_1)  -\nabla_x \cdot (\rho f_2) = (\L_1+\L_2) \rho.$ 
\end{proof}

The Koopman operator $\K^t : L^\infty(\XX) \mapsto L^\infty(\XX)$ propagates observable functions forward in time along trajectories of the system and is defined as
\begin{equation}
    [\K^th](x) = [h\circ\Phi^t](x)
\end{equation}
where $h(x)$ is an observable. The Koopman and Perron-Frobenius operators are adjoint to one another, 
\begin{equation} \label{eq:adjoint}
    \int_{\XX}[\K^th](x)\rho(x)dx = \int_{\XX}h(x)[\P^t\rho](x) dx \,.
\end{equation}


\begin{figure*} 
    \centering
    \includegraphics[width=0.19\linewidth]{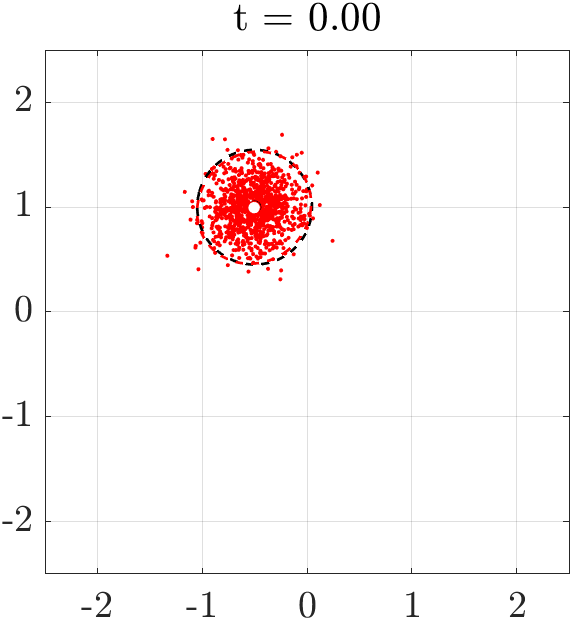}
    \includegraphics[width=0.19\linewidth]{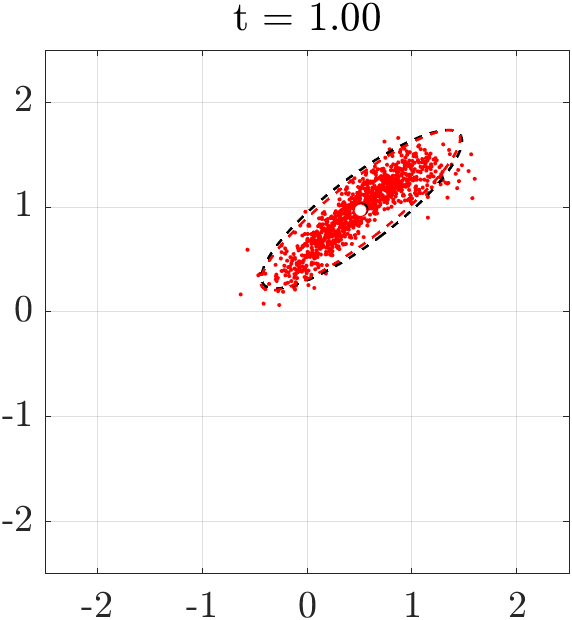}
    \includegraphics[width=0.19\linewidth]{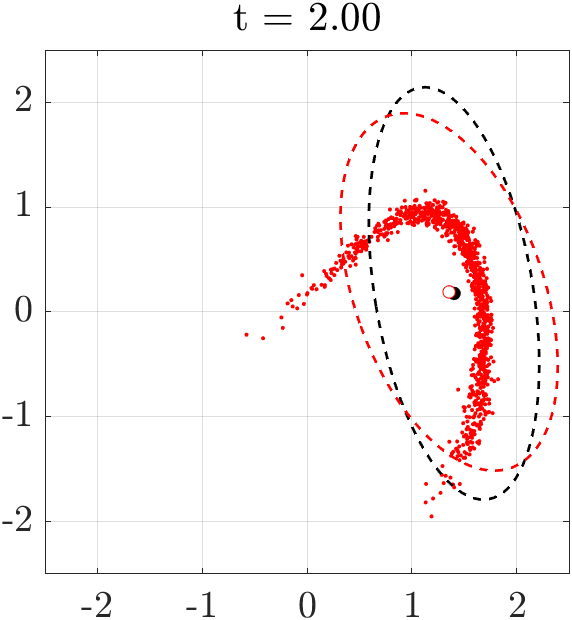}
    \includegraphics[width=0.19\linewidth]{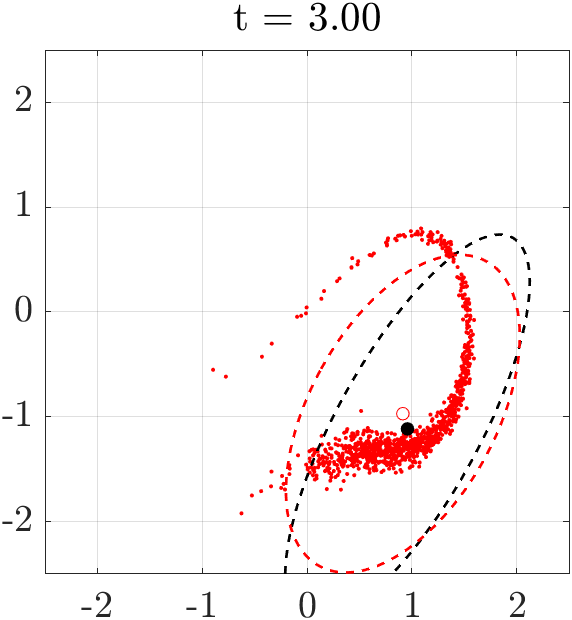}
    \includegraphics[width=0.19\linewidth]{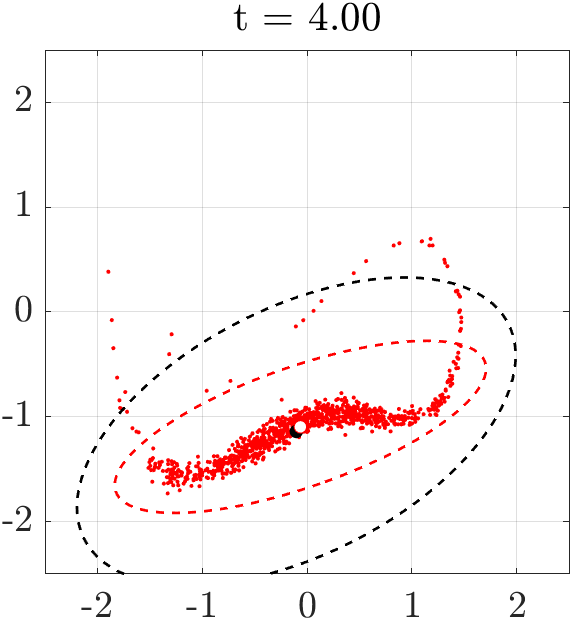}
    \caption{Moment propagation of proposed method for a Duffing oscillator with sinusoidal forcing.  Red points show trajectories from initial conditions sampled from the initial density, $\rho(x(0))\sim\N([-0.5;1],0.05I)$. Red circle and red ellipse show the sample mean and 2$\sigma$ sample covariance ellipse, respectively.  The black circle and black ellipse are the predicted mean and $2\sigma$ covariance ellipse. }
    \label{fig:duffing_ol_snapshots}
\end{figure*}

\section{Numerical approximation of the Perron-Frobenius operator and generator }
Here, we implement extended dynamic mode decomposition (EDMD) \cite{williams2015data} for the computation of the Perron-Frobenius operator, which we outline below, largely following \cite{Klus16_onthenumerical}.  

We begin by selecting a dictionary $\DD$ of $k$ scalar-valued basis functions, $\DD = \{\psi_1, \psi_2, \dots, \psi_k\}$, where $\psi_i:\XX\mapsto\RR$ for $i = 1, \dots, k$, and denote by $\VV$ the function space spanned by the elements of $\DD$. We then collect trajectory data with fixed timestep, $\Delta t$, arranged into snapshot matrices as 
\begin{align}
X &= 
\begin{bmatrix} \label{eq:snapshotx}
x_1 & , ~\cdots~ , & x_m 
\end{bmatrix}\\
Y & = 
\begin{bmatrix} 
x_1^+ &, ~\cdots~, & x_m^+ 
\end{bmatrix}
\end{align}
where the subscript $i=1,\dots, m$ is a measurement index and $x_i^+ = \Phi^{\Delta t}(x_i)$. 

We then approximate the observable function $h$ and density $\rho$ in Eq. \ref{eq:adjoint} by their projections onto $\VV$.  That is, 
\begin{align}
    h(x) &\approx \hat{h}^T\Psi(x)\\
     \rho(x) &\approx \Psi^T(x) \hat{\rho} \label{eq:dens_proj} 
\end{align}
where $\hat{h},\,\hat{\rho}\in\RR^k$ are column vectors containing the projection coefficients and $\Psi:\XX\mapsto\RR^k$ is a column-vector valued function where the elements are given by $[\Psi(x)]_i = \psi_i(x)$. 
Substituting these expansions into Eq. \ref{eq:adjoint}, we have 
\begin{equation}
    \int_{\XX} \K^{\Delta t}[\hat{h}^T\Psi]\Psi^T\hat{\rho} \, dx 
    = 
    \int_{\XX} \hat{h}^T\Psi \P^{\Delta t}[\Psi^T\hat{\rho}]\,dx\,.
    \label{eq:adj_expanded}
\end{equation}
Then replacing $[\K^{\Delta t}\Psi](x) = \Psi(x^+)$ and assuming that $\P^{\Delta t}$ can be approximated by a matrix $P$ operating on the coordinates $\hat{\rho}$
gives a least-squares problem for the matrix $P$
\begin{equation}
    \min_P \|\Psi_Y\Psi_X^T - \Psi_X\Psi_X^TP\|_2^2
\end{equation}
where $\Psi_X$,$\Psi_Y \in\RR^{k\times m}$ are matrices with columns formed by evaluating $\Psi$ on the columns of $X$ and $Y$ respectively.  The analytical solution of this least squares problem is 
\begin{equation}
    P = \left(\Psi_X\Psi_X^T\right)^\dagger\Psi_Y\Psi_X^T 
\end{equation}
where $(\cdot)^\dagger$ is the Moore-Penrose pseudoinverse. 

Given this matrix approximation of the operator, $P$, if the timestep $\Delta t$ chosen in the data collection is sufficiently small, the corresponding matrix approximation $L$ of the Perron Frobenius generator can be approximated based on the limit definition of the generator in Eq. \ref{eq:generator_limit}. as 
\begin{equation}
    L\approx \frac{P - I_k}{\Delta t}
\end{equation}
where $I_k$ is the $k\times k$ identity matrix.  Just as the matrix operator $P$ approximates the propagation of a density function $\rho$ by advancing the projection coordinates $\hat{\rho}$ forward for a finite time, the approximation of the generator allows us to approximate the infinitesimal action of the operator $\P^t$ by approximating the time derivative of the projection coordinates 
\begin{equation}
    \frac{d\hat{\rho}}{dt} = L\hat{\rho}\,.
\end{equation}

\subsection{Extension to controlled systems} \label{sec:gen_control}
In the context of applying the Koopman operator to control systems, several recent works have noted the usefulness of formulating the problem in terms of the Koopman generator, rather than the Koopman operator itself  (see \cite{goswami2017globalbilin,klus2020gedmd,rowley2020interpolated} and others), 
which typically results in a lifted system that is bilinear in the control and lifted state. This approach allows for a better approximation of the effects of control, especially for systems in control-affine form 
\begin{equation}
    \dot{x} = f(x) + \sum_{i = 1}^{n_c}g_i(x)u_i
\end{equation}
as it expresses the effect of the control vector fields $g_i$ in a way that is also dependent on the lifted state.  
Here we apply a similar approach to the density transport problem, expressed in terms of the Perron-Frobenius generator. As shown in Ref. \cite{rowley2020interpolated} for the Koopman generator, by the property of the Perron-Frobenius generator given in Lemma \ref{lem_additivity}, if the dynamics are control-affine, then the generators are also control affine, as can be seen by application of Eq. \ref{eq:generator_continuity}.  This leads to density transport dynamics of the following form
\begin{equation}\label{eq:dens_dyn}
    \frac{d}{dt}\rho(x) = (\L_0\rho)(x) + \sum_{i=1}^{n_c}u_i(\B_i\rho)(x)
\end{equation}
where $\L_0$ is the Perron Frobenius generator associated with the vector field $f(x)$ and similarly, the $\B_i$ are the Perron Frobenius generators associated with the control vector fields $g_i(x)$. 
Therefore, given the finite dimensional approximation of these generators, we can approximate the density transport dynamics as 
\begin{equation} \label{eq:dens_dyn_fin}
    \frac{d\hat{\rho}}{dt} = L_0\hat{\rho} + \sum_{i=1}^{n_c} u_iB_i\hat{\rho}
\end{equation}
where the matrices $L_0$ and $B_i$ are the matrix approximations of the operators in Eq. \ref{eq:dens_dyn}.


\subsection{Propagation of moments} \label{sec:moments}
In order to make the problem of controlling a density function using a finite dimensional control input well-posed, we formulate the problem as a control problem of a finite number of outputs. In particular, we will describe the density function, in terms of a finite number of its moments. 
For a scalar $x$, recall that the $i^{\text{th}}$ raw moment, $m_i$ is defined as \(m_i=\int_\XX x^i \rho(x) dx \) and the $i^{\text{th}}$ central moment, $\mu_i$ about the mean $m_1$ is \(\mu_i=\int_\XX (x-m_1)^i \rho(x) dx\).

Given a projection of the density as in Eq. \ref{eq:dens_proj}, the mean is approximated as 
\begin{equation}\label{eq:m1}
    m_1^i = \int x_i \rho(x) dx =\hat{\rho}^T \int x_i\Psi(x) dx 
\end{equation}
which simply indicates that the mean of the density can be written as a summation of the means of the dictionary functions, weighted by the projection coefficients. For higher moments, if the central moment is considered, it will be polynomial in the projection coefficients due to its dependence on the mean, whereas the raw moments remain linear in the projection coefficients. For this reason, we choose to work with 
the raw moments, as the central moments can also be expressed in terms of the raw moments. 
%
Similarly, the second raw moment can be written as 
\begin{equation}\label{eq:m2}
m_2^{ij} = \int x^ix^j\rho(x)dx = \hat{\rho}^T\int x^ix^j\Psi(x)dx
\end{equation}
where, again, superscripts $i$ and $j$ are coordinate indices. 

\begin{figure*}
    \centering
    \includegraphics[width=0.19\linewidth]{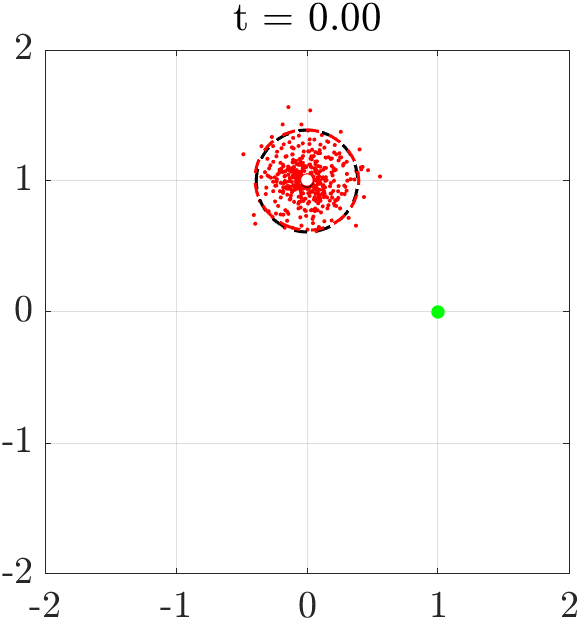}
    \includegraphics[width=0.19\linewidth]{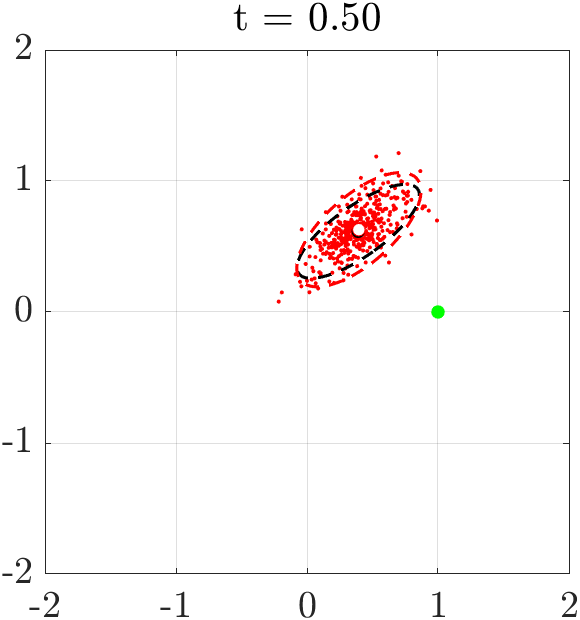}
    \includegraphics[width=0.19\linewidth]{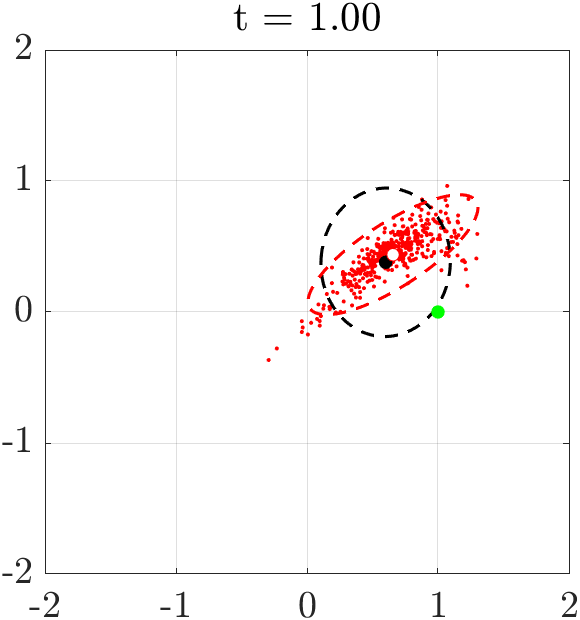}
    \includegraphics[width=0.19\linewidth]{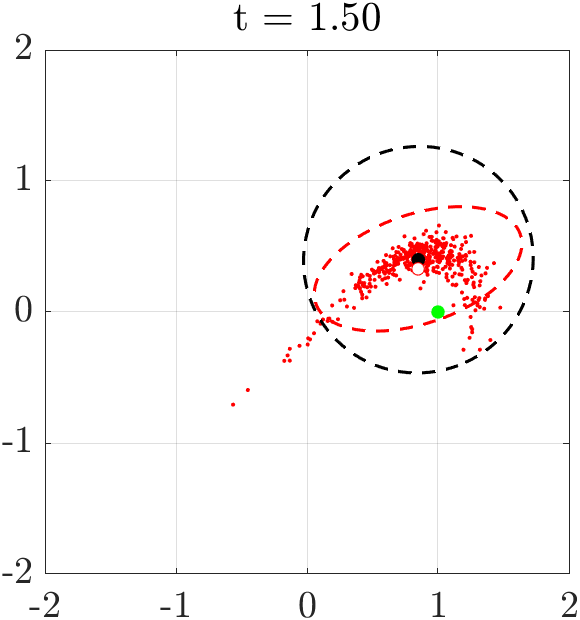}
    \includegraphics[width=0.19\linewidth]{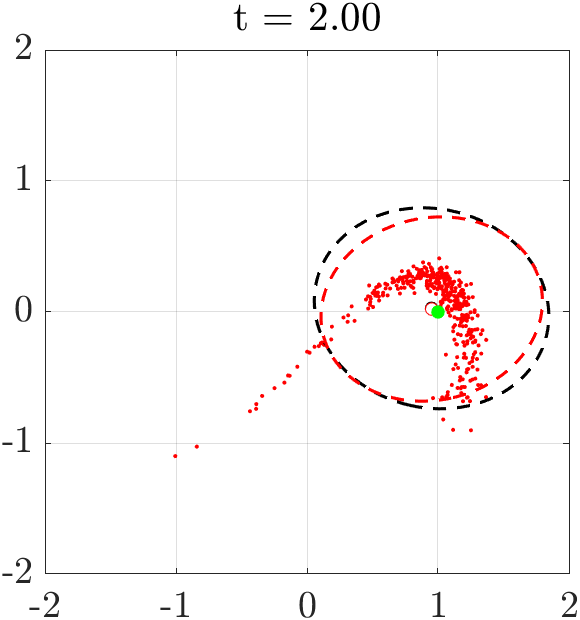}
\caption{Controlled transport through the Duffing system from an initial density $\rho(x(0))\sim\N([0;1],0.025I)$} toward the equilibrium point at $(0,1)$ (green circle). Red circle and red ellipse show the sample mean and 2$\sigma$ sample covariance ellipse, respectively.  The black circle and black ellipse are the predicted mean and $2\sigma$ covariance ellipse. 
    \label{fig:duffing_control_snapshots}
\end{figure*}

\begin{figure}
    \centering
    \includegraphics[width=\linewidth]{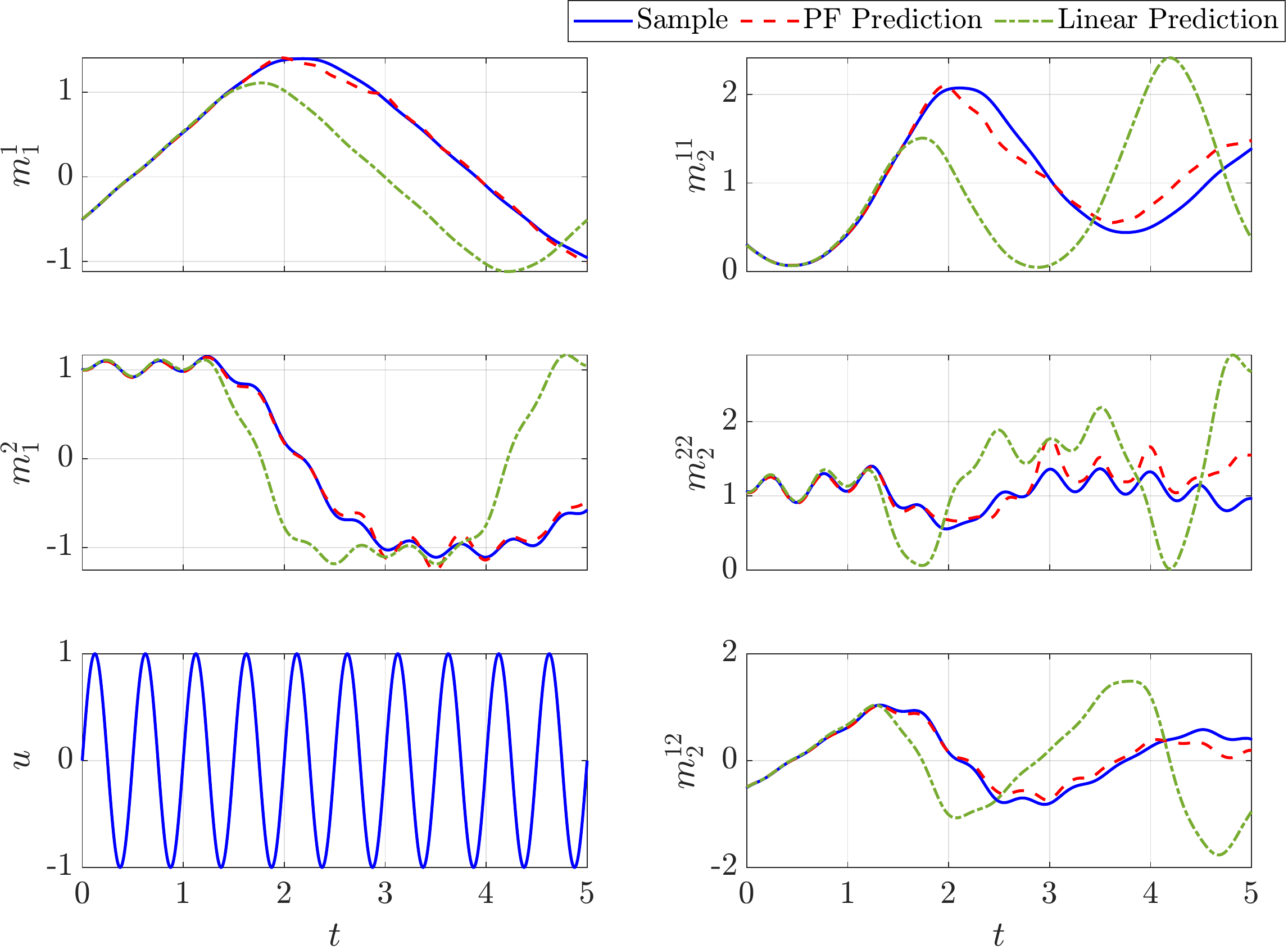}
    \caption{Moment propagation for a forced Duffing oscillator.  Left, top: First raw moment (mean).  Left, bottom: sinusoidal control signal. Right: 2nd raw moment. 
        }
    \label{fig:duffing_ol_moments}
\end{figure}
\subsection{Numerical example} \label{sec:example_prop}
To illustrate the ability of the proposed framework to propagate density functions forward in time, we consider the propagation of an initial density for a forced Duffing oscillator system, given by 
\begin{equation} \label{eq:duffing}
    \frac{d}{dt}\begin{pmatrix}
        x_1\\x_2
    \end{pmatrix}
     = 
  \begin{pmatrix}
         x_2\\
         x_1 - x_1^3 + u
     \end{pmatrix}
\end{equation}
where $u$ is the control input.  For the purpose of this simulation, we set $u(t) = \sin(4\pi t)$ and the prediction results are shown in Figs. \ref{fig:duffing_ol_snapshots} and \ref{fig:duffing_ol_moments}.  For the generator calculation, a dictionary of Gaussian radial basis functions is used where the centers lie on an evenly spaced $30\times 30$ grid ranging from $-2.5$ to $2.5$ in $x_1$ and $x_2$.  The operators are approximated using data collected from short time trajectories with $\Delta t=0.005$ for a $50\times 50$ grid of initial conditions on the same region. The predicted moment is compared to the sample moment obtained from $1000$ trajectories from initial conditions sampled according to the initial density. 
We see that the moment propagation of the proposed method is good for at least 3 seconds, which motivates the use of this method in a control formulation, as detailed in the following sections. 
Also shown for comparison in Fig. \ref{fig:duffing_ol_moments} is a linear prediction, which is computed by propagating the initial Gaussian through a linearization of Eq. \ref{eq:duffing}, where the linearization is re-computed at each timestep about the predicted mean, as is commonly done in the a priori prediction step of an extended Kalman filter. 


\section{Control formulation }
We have shown in Sec. \ref{sec:gen_control} and \ref{sec:moments} that the problem of steering a density $\rho$ to a desired density can be expressed as an output tracking problem on a lifted, bilinear system given by Eq. \ref{eq:dens_dyn_fin}, where the projection coefficients $\hat{\rho}$ can be interpreted as the lifted state. 
Then, if the raw moments are taken to be the relevant output, the output, $y$ is linear in the lifted state, $y = C\hat{\rho}$, where the elements of the output matrix $C$ are given by rewriting Eqs. \ref{eq:m1}, \ref{eq:m2} in matrix form. 

For the optimal output tracking problem, we consider a discrete time optimal control problem 
\begin{subequations}
\begin{align}
    \min_{u_1,u_2,\dots,u_{H-1}} &\sum_{t=1}^{H-1}l(\rh_t,u_t) + l_H(\rh_H)\\[1ex]
    \mathrm{s.t.} \qquad & \hat{\rho}_{t+1} = F(\hat{\rho}_t,u_t) \label{eq:disc_dyn} \\
    & y_t = C\rh_t
\end{align}
\end{subequations}
where $H$ is the number of timesteps in the time horizon and Eq. \ref{eq:disc_dyn} represents the discrete time version of Eq. \ref{eq:dens_dyn_fin}.

In particular, for output tracking, we consider a quadratic cost of the form
    \begin{align}
        l(\rh_t,u_t) &= (y_t - y^{\mathrm{ref}}_{t})^TS(y_t - y^{\mathrm{ref}}_{t}) + u_t^TRu_t 
        \label{eq:cost_x}
        \\
        l_H(\rh_H) &= (y_H - y^{\mathrm{ref}}_{H})^TS_H(y_H - y^{\mathrm{ref}}_{H}) 
    \end{align}
where $S$ and $R$ are weighting matrices which define the relative penalty on tracking error and control effort, respectively. This cost is quadratic in $\rh_t$ (with a linear term).  

It is well known that for optimal control problems on bilinear systems with quadratic cost, an effective way of solving the problem is by iteratively linearizing and solving a finite time linear quadratic regulator (LQR) problem about a nominal trajectory, utilizing the Ricatti formulation of that problem \cite{hofer1988iterative}.  For this reason, we solve the optimal control problem using differential dynamic programming (DDP) \cite{tassa2012synthesis,yakowitz1984computational}, which is closely related to the method of iterative LQR.  

\begin{figure*}
    \centering
    \includegraphics[width=\linewidth]{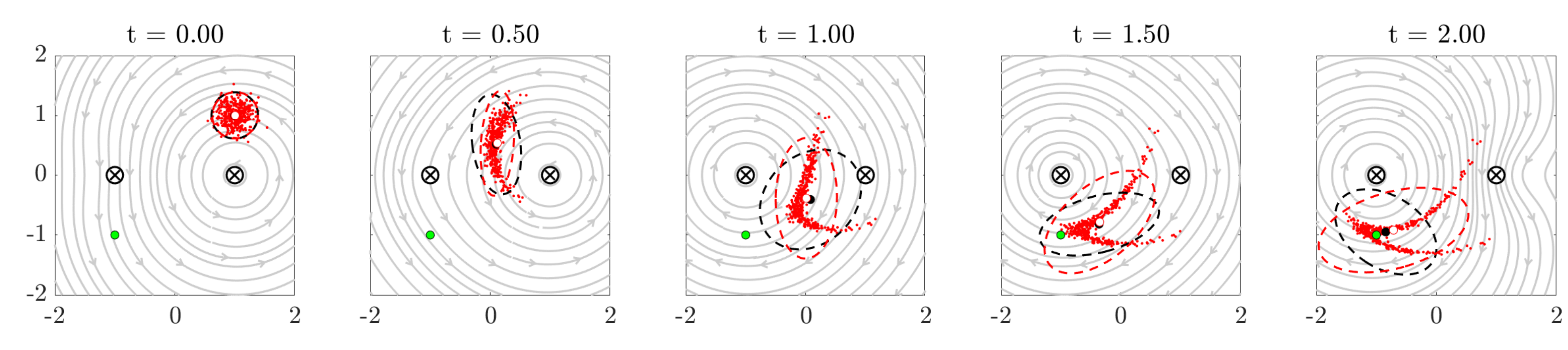}
    \caption{Controlled transport of fluid particles, driven by two rotlets, or micro-rotors, in a Stokes flow from an initial density, $\rho(x(0))\sim\N([1;1],0.025I)$ toward a target mean (green circle) at $(-1,-1)$. The rotors are located at $(-1,0)$ and $(0,1)$, as indicated by the black circle-cross.  Red circle and red ellipse show the sample mean and 2$\sigma$ sample covariance ellipse, respectively.  The black circle and black ellipse are the predicted mean and $2\sigma$ covariance ellipse.  Gray streamlines indicate the flow field produced by the rotlets at the instant shown. }
    \label{fig:rotlet_snapshots}
\end{figure*}

\section{Validation and Results}
Here we consider two examples of the density control formulation using Perron Frobenius generators.   
\subsection{Example 1: Forced Duffing oscillator}  

\begin{figure}
    \centering
    \begin{minipage}[c]{0.49\linewidth}
    \includegraphics[width=\linewidth]{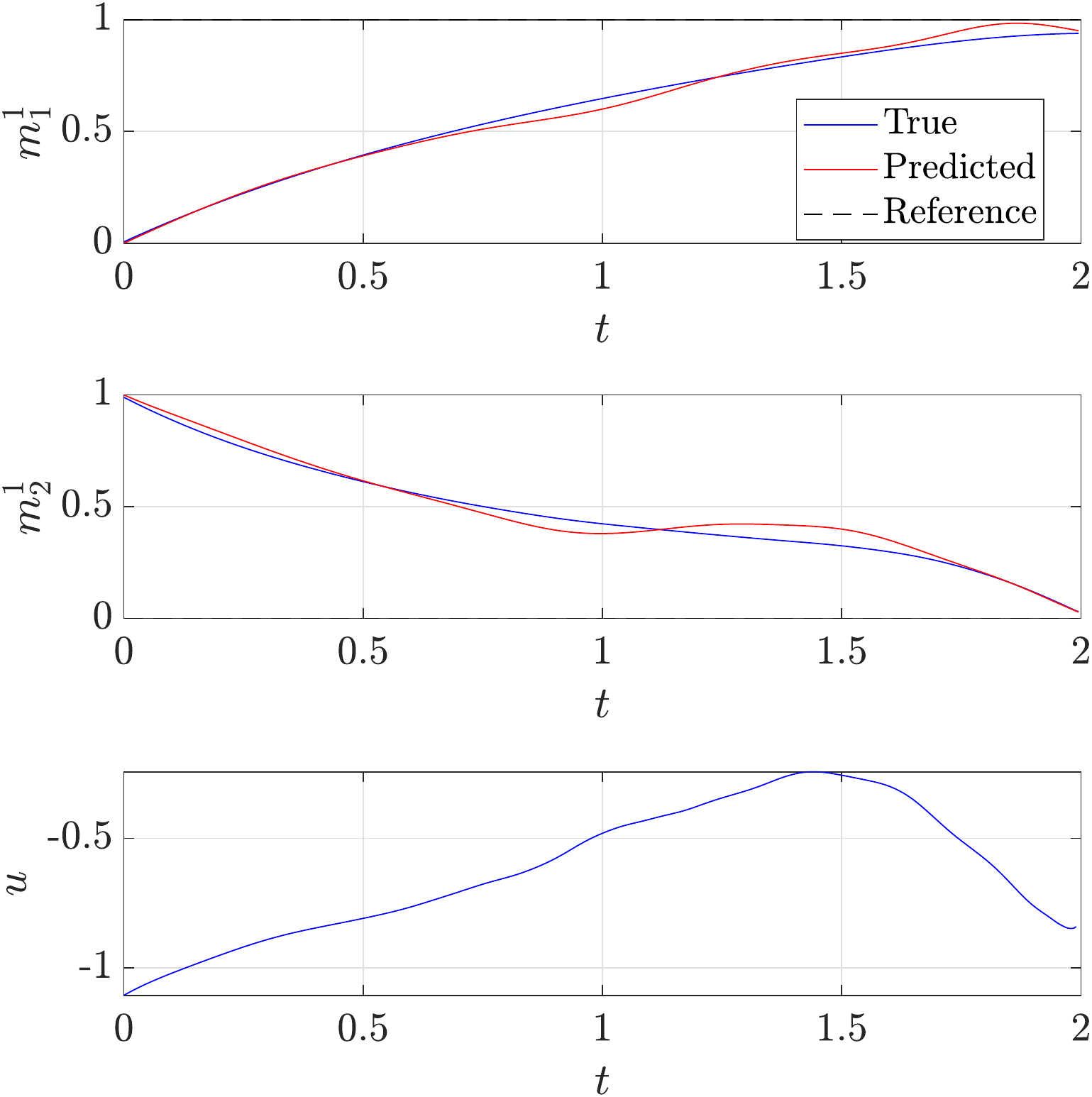}
    \end{minipage}
    \begin{minipage}[c]{0.49\linewidth}
    \includegraphics[width=\linewidth]{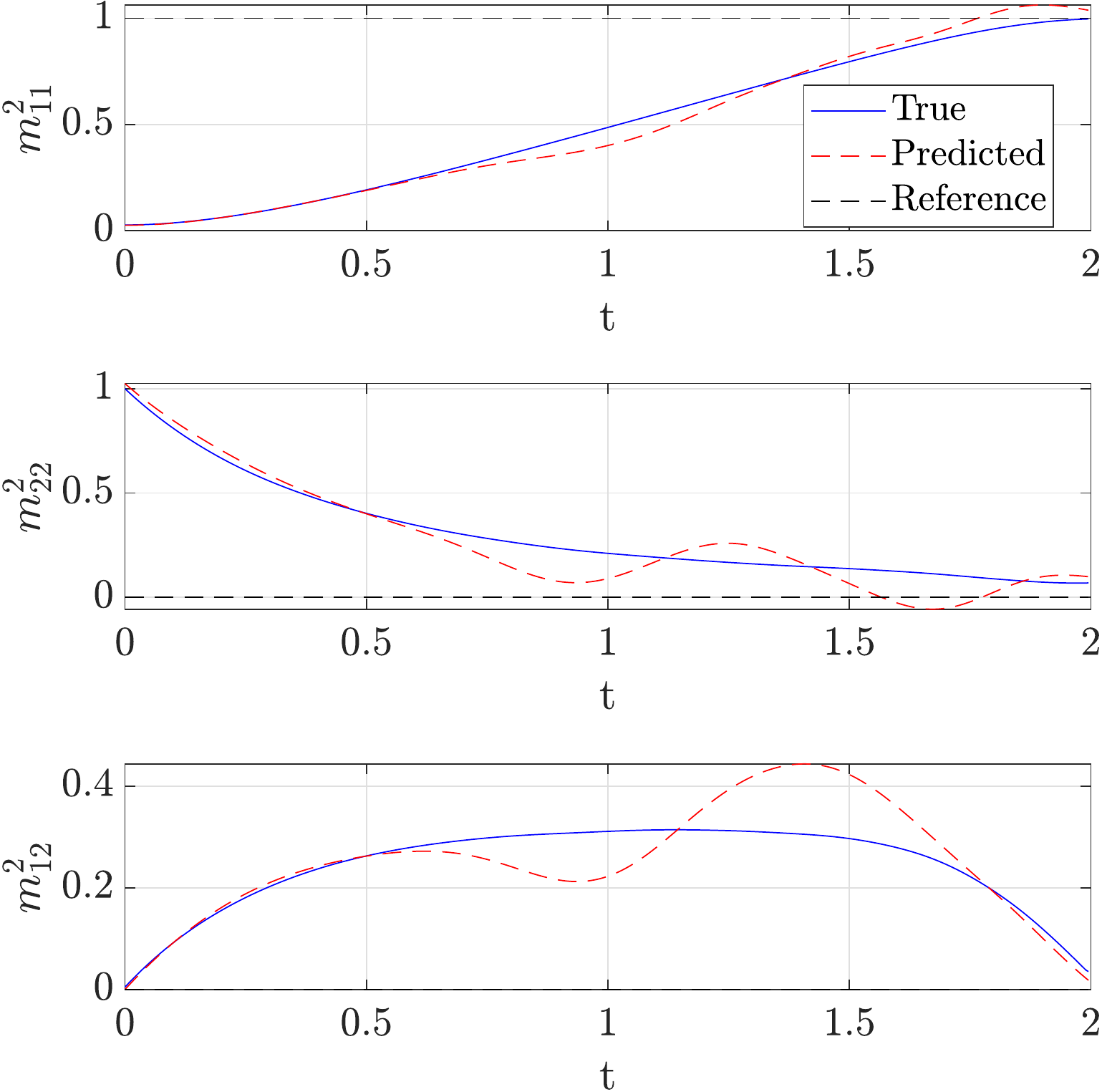}
    \end{minipage}
    \caption{Control of raw moments for the forced Duffing system. The black line indicates the target (reference).  Left, bottom: control from differential dynamic programming for the generator system. Right: 2nd raw moment.  Predicted values are from the Perron-Frobenius generator computation.  True values are given by the sample moment. }
    \label{fig:duffing_control_moments}
\end{figure}

\begin{figure}
    \centering
    \includegraphics[width=\linewidth]{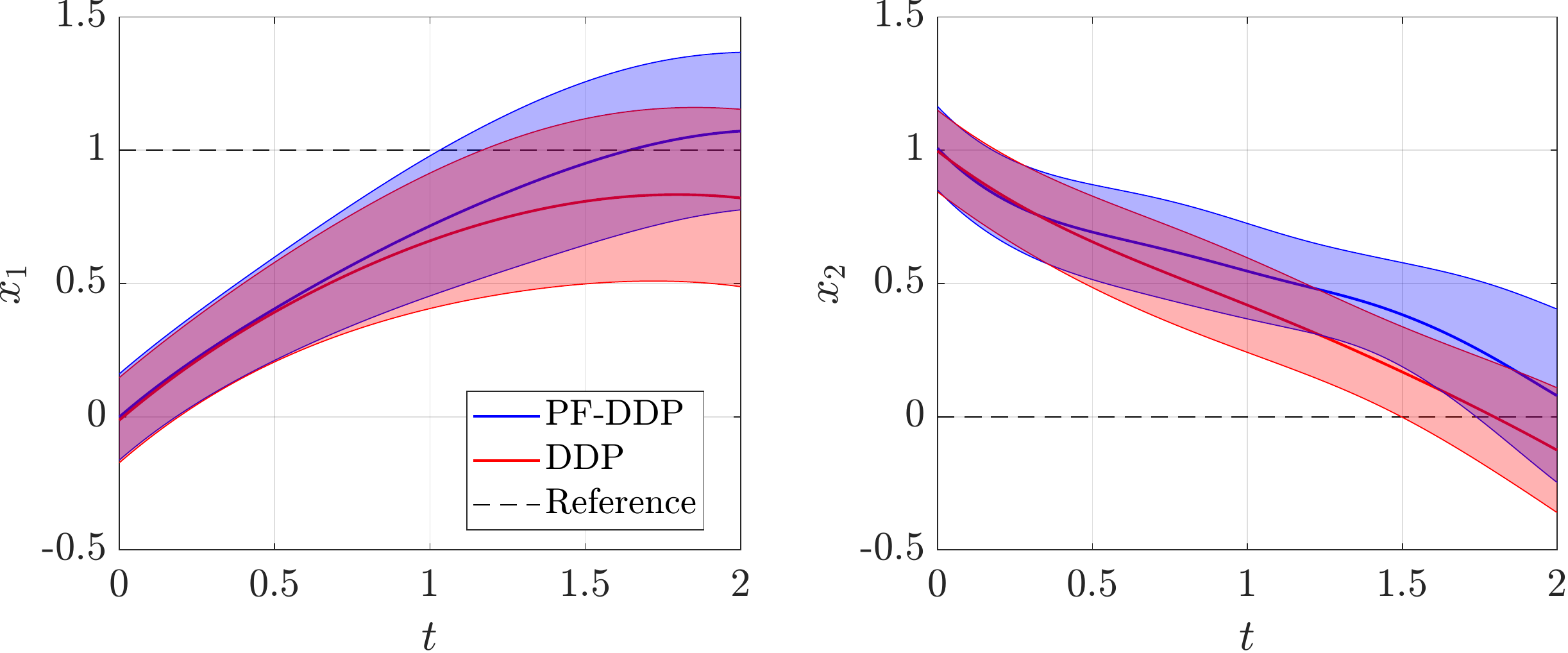}
    \caption{Comparison of the proposed Perron Frobenius generator DDP with standard DDP for the forced Duffing system.  Shaded region indicates one standard deviation }
    \label{fig:pf_ddp_comp}
    \vspace{-1em}
\end{figure}

As a first example, we consider the forced Duffing oscillator of Eq. \ref{eq:duffing}, and we use DDP to determine a control sequence to steer a Gaussian initial density $\rho(x(0))\sim\N([1;1],0.025I)$ toward the equilibrium point at $(1,0)$ over a time horizon of 2s.  The generators are approximated using the same data as described in Sec. \ref{sec:example_prop} and the same timestep of $0.005$s is used for DDP. In differential dynamic programming, the in-horizon cost weights on the reference error of the first two moments and control effort are all set to unity. The terminal cost on the reference error in the moments is set to $1000$, and only the first two moments are considered. The target raw second moment is computed from the desired mean, with the desired variance taken to be zero. The results of this computation are shown in Fig. \ref{fig:duffing_control_snapshots}-\ref{fig:duffing_control_moments}.

In Fig. \ref{fig:pf_ddp_comp}, we show a comparison of the performance of the proposed controller, labelled `PF-DDP' with a standard DDP controller, which computes a control on the Duffing system directly (rather than a lifted state), with the initial condition being the mean.  That is, the standard DDP controller acts on the mean as if it were a deterministic initial condition of the system. 
The comparison shown is the mean and one standard deviation region from trajectories from a set of 500 initial conditions sampled from the initial distribution, with each of the respective controllers applied.  
We see that the proposed controller moves the mean of the distribution closer to the reference, while maintaining nearly the same variance as the standard DDP controller. 

\subsection{Example 2: Rotor-driven Stokes flow }
As a second example, we consider the problem of steering a distribution of fluid particles in a Stokes flow, where the flow is produced by two micro-rotors.  The micro-rotors are modelled as rotlets, the Stokes flow singularity associated with a point torque in the fluid.  
For a collection of $N_r$ rotlets, this flow is given by 
\begin{equation}
    \frac{d \mathbf{x}}{dt} = \sum_{i=1}^{N_r} \frac{T_i\times (\mathbf{x} - \mathbf{\bar{x}_i})}{\|\mathbf{x} - \mathbf{\bar{x}}_i\|^3}
\end{equation}
where $\mathbf{\bar{x}}_i$ and $T_i$ denote the location and torque of the $i$-th rotor, respectively. We consider the case where there are two such rotors lying in the $x_1$-$x_2$ plane, located at $\mathbf{\bar{x}_1}=(-1,0)$ and $\mathbf{\bar{x}_2}=(1,0)$ and the controls for the problem are taken to be torques $u_1 = T_1$, $u_2=T_2$, where the direction of these torques is taken to be normal to the $x_1$-$x_2$ plane (in the positive $x_3$ direction. 

For this example, the task is to drive a Gaussian initial density 
$\rho(x(0))\sim\N([1;1],0.025I)$ toward a target mean at $(-1,-1)$ over a time horizon of 2s. Again, we consider a timestep of $\Delta t = 0.005$s for both the computation of the generators and for DDP. The in horizon costs weights for the mean error are set to 2, while the weights for the second moment error and control effort are unity. The terminal cost weight on the mean error is 1000 and the terminal cost weight on the second moment error is 500. Results from this computation are shown in Figs. \ref{fig:rotlet_snapshots} and \ref{fig:rotlet_moments}. 

We see in Fig. \ref{fig:rotlet_moments} that DDP yields a control which drives the mean to the target by first giving a significant counterclockwise torque on the right rotor to drive the density to a point between the rotors, at which point the left rotor is initiated to drive the flow with a clockwise torque, pulling the density mean near to the target. 

This example demonstrates an alternative physical interpretation of the density transport problem, where the density represents a distribution of fluid particles.  This also demonstrates the effectiveness of the proposed method on a system which is linear in the controls, but in which the control vector fields are  nonlinear.

\begin{figure}
    \centering
    \begin{minipage}[c]{0.49\linewidth}
    \includegraphics[width=\linewidth]{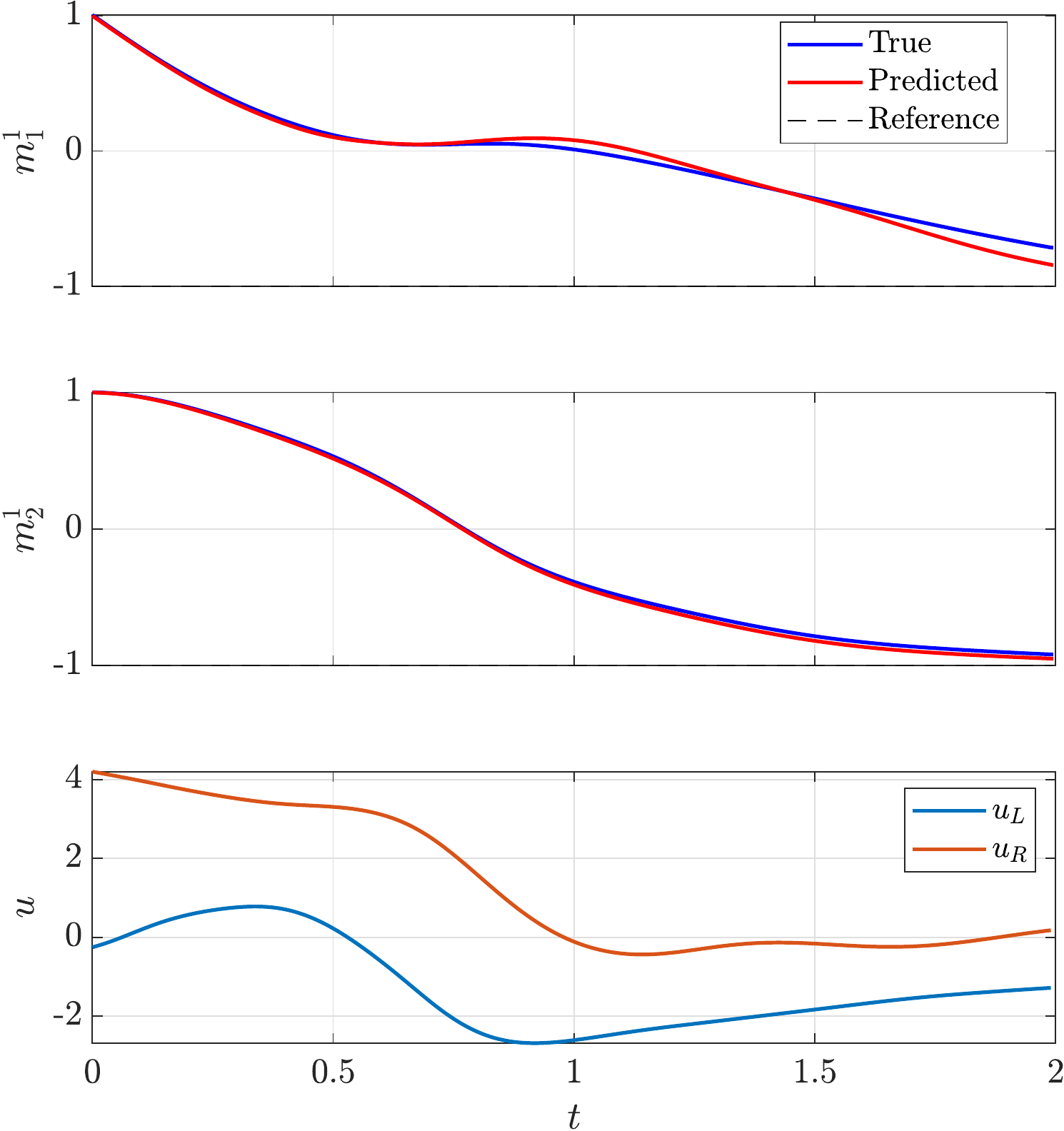}
    \end{minipage}
    \begin{minipage}[c]{0.49\linewidth}
    \includegraphics[width=\linewidth]{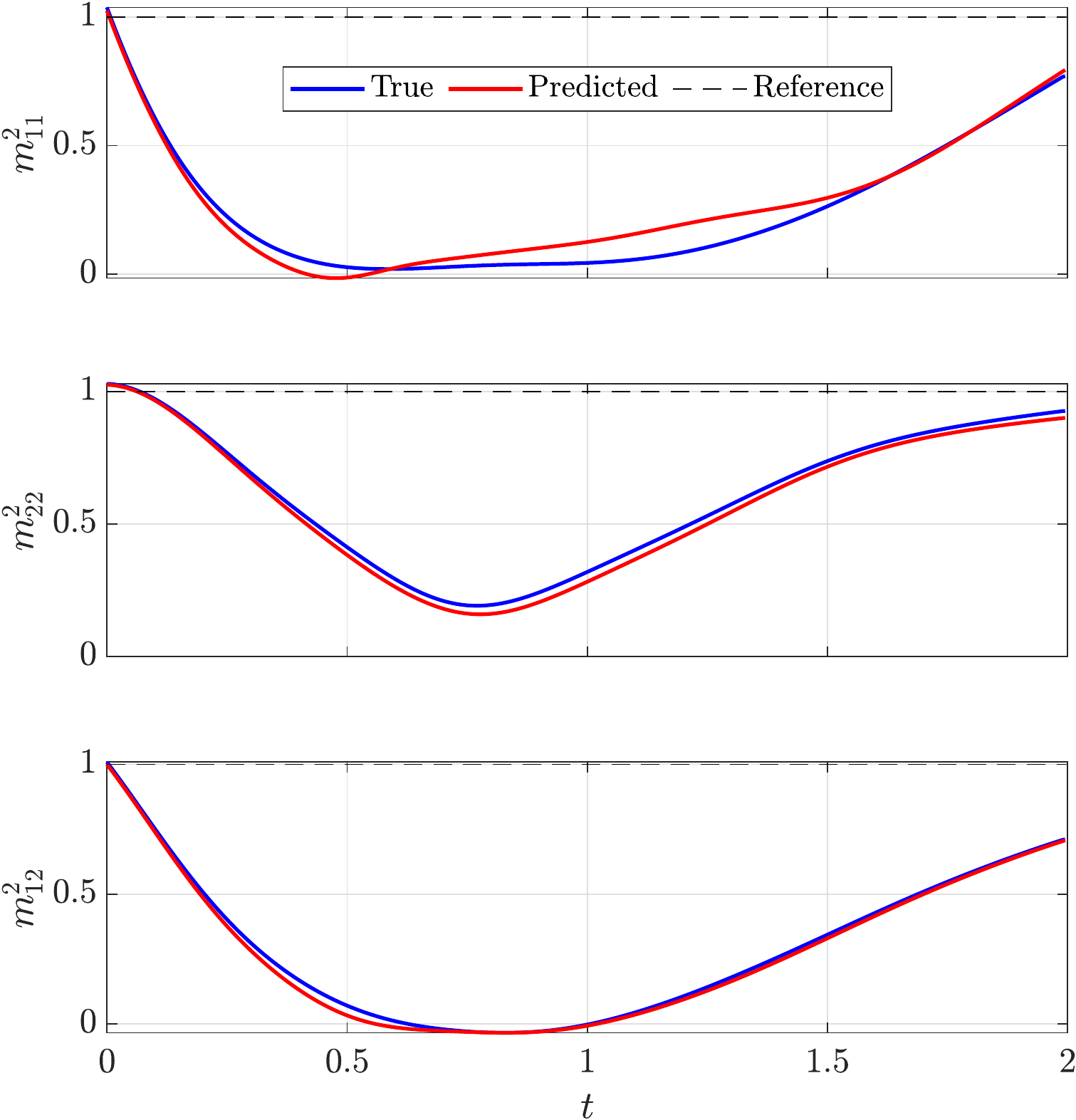}
    \end{minipage}
    \caption{Control of raw moments for the rotlet-driven Stokes flow. The black line indicates the target (reference). Left, top: 1st raw moment. Left, bottom: control inputs $u_L$ and $u_R$ for the left and right rotor.  Right: 2nd raw moment.  Predicted values are from the Perron-Frobenius generator computation.  
    True values are given by the sample moment.
    }
    \label{fig:rotlet_moments}
    \vspace{-1em}
\end{figure}



\section{Conclusion}
In this work, we have studied the problem of transporting density functions of states through a controlled dynamical system.  This problem formulation has applications both in fluid mechanics and in the control of uncertain systems.  Our approach is based on approximations of the Perron-Frobenius operator, whereby we show that approximations of this operator and its generator can be used to model the density transport dynamics as a high-dimensional system which is bilinear in the lifted state and the control. 
We demonstrated this approach on two examples, a forced Duffing system, in which the density can have the interpretation as an uncertainty in the initial state and on a rotor driven Stokes flow, in which the density formulation takes on the fluid mechanics interpretation of describing a distribution of fluid particles. Future work in these areas could include extending the proposed control formulation for use in a constrained model predictive control framework for uncertain systems or by studying the fluid transport by more realistic biological microswimmers or artificial microrobots \cite{bt_acc2020}. 

\bibliographystyle{ieeetr}
\bibliography{densitycontrol}
\end{document}